\documentclass[11pt]{llncs}

\usepackage[latin5]{inputenc}

\usepackage{amssymb}
\usepackage{amsmath}
\usepackage{textcomp}
\usepackage{array}
\usepackage{multirow}
\usepackage{color}
\usepackage[latin5]{inputenc}

\title{Computation with narrow CTCs\thanks{This 
work was partially supported by the Scientific and Technological
Research Council of Turkey (T\"{U}B\.ITAK) with grant 108E142.}}
\author{A. C. Cem Say\ \and  Abuzer Yakary{\i}lmaz}
\institute{Bo\u{g}azi\c{c}i University, Department of Computer Engineering,\\ Bebek 34342 \.{I}stanbul, Turkey \\
\email{{say,abuzer}@boun.edu.tr}
 \\~~\\
\today
}

\begin{document}

\newlength{\twidth}
\maketitle
\pagenumbering{arabic}

\begin{abstract} \label{abstract:Abstract}
We examine  some variants of computation with closed timelike curves (CTCs), where various restrictions are imposed on the memory  of the computer,  and the information carrying capacity and range of the CTC. We give full characterizations of the classes of languages recognized by polynomial time probabilistic and quantum computers that can send a single classical bit to their own past. Such narrow CTCs are demonstrated to add the power of limited nondeterminism to deterministic computers, and lead to exponential speedup in constant-space probabilistic and quantum computation. We show that, given a time machine with constant negative delay, one can implement CTC-based computations without the need to know about the runtime beforehand.
\end{abstract}

\section{Introduction} \label{section:Introduction}
It is known \cite{AW09} that adding the capability of sending a polynomial number of bits through a closed timelike curve (CTC) to the past, so that it can be used as part of the input, to models as weak as constant-depth, polynomial-size Boolean circuits increases their computational power significantly, to match that of polynomial-space Turing machines. Interestingly, adding the same capability to a polynomial-time quantum computer results once again in the ability to solve precisely the problems in $\mathsf{PSPACE}$, leading to Aaronson and Watrous' conclusion \cite{AW09} that ``CTCs make quantum and classical computing equivalent."

Since the information carrying capacity (\textit{width}) of any CTC is finite, and the cost of building such a channel to the past may depend on its width critically, it is important to examine the power of computational models with ``narrow" CTCs, i.e. those which are restricted to use a single-bit channel, regardless of the size of the input. The study of narrow CTCs \cite{B04,AW09} has focused on polynomial-time computers as the core model until now, with results showing that classical computers augmented with narrow CTCs can recognize any language in $\mathsf{BPP_{path}}$, whereas quantum computers with this capability can solve all problems in $\mathsf{PP}$. Since it is not known whether any of the containments in $\mathsf{P} \subseteq \mathsf{NP} \subseteq \mathsf{BPP_{path}} \subseteq \mathsf{PP} \subseteq \mathsf{PSPACE}$ is proper, we cannot presently say whether narrow CTCs are ``useful" in these setups at all, and if so, whether they confer the same amount of power to classical and quantum models.

In this paper, we will complement the results mentioned above to provide full characterizations of the classes of languages recognized by probabilistic and quantum computers with narrow CTCs. It turns out that a CTC with a single classical bit provides a computer with precisely the power of postselection \cite{Aa05}.

The effects of restrictions on the core model of computation to be used inside the CTC on the power of the resulting setup were posed as open questions in \cite{AW09}. In this regard, we consider real-time probabilistic and quantum finite automata (PFAs and QFAs), as well as several deterministic automaton models, as possible core models to be augmented with a narrow CTC. We show that narrow CTCs add  the powers of limited nondeterminism or a certain kind of two-wayness to real-time automata. We are therefore able to prove that real-time PFAs, QFAs and deterministic pushdown automata (DPDAs) with narrow CTCs are strictly more powerful than their standard counterparts, whereas this addition does not increase the power of one-way deterministic finite automata (DFAs), or any deterministic two-way model. We also show that QFAs with narrow CTCs outperform their probabilistic counterparts, in contrast to the Aaronson-Watrous result on polynomial-time computation.

The rest of this paper is structured as follows. Section
\ref{section:Preliminaries} summarizes previous work on the use of
CTCs in computation. We define an alternative model of augmenting
computers with CTCs in Section \ref{section:ourmodel}. Section
\ref{section:pfaqfactc} establishes the equivalence of the powers of
postselection and narrow CTCs. In \ref{section:detctc}, we examine the
effects of endowing weak deterministic machines with narrow CTCs.
In Section
\ref{section:feasibility}, we demonstrate how our model of Section
\ref{section:ourmodel} can be implemented with time machines with
constant negative delay. Section \ref{section:Conclusion} is a conclusion.

\section{Preliminaries} \label{section:Preliminaries}
The existence of CTCs does  not seem to be incompatible with the best available theories of spacetime. To cite just one example, an influential paper \cite{MTY88} by Morris, Thorne, and Yurtsever describes how a technologically advanced civilization could first create a ``wormhole", and then transform it to a ``time machine" that can be used to send messages, or even people, backwards in time. The  time machines of \cite{MTY88} cannot be used to send anything to a time before their date of construction, say, $d$, and have a constant ``range" of, say, $T$ seconds, (determined by their builder at the beginning,) such that a message sent at any time point $t$ ($t>d+T$) is received at time point $t-T$. Once built, the machine can be used as many times as one wishes for such transmissions.

As noted by many authors of science fiction, a major problem with time travel is the ``Grandfather Paradox," where a time traveler from the future prevents himself from traveling in the first place, leading to confusion about the state of the universe at the presumed time, say, $A$, of his arrival: He arrives if and only if he does not arrive. It was thought that Nature would prevent this logical inconsistency by simply not allowing time travel scenarios of that kind to be realized. Note that this argument assumes that the universe is supposed to be in exactly one, deterministic state, at all times. Probabilistic and quantum theories do not include this restriction, and David Deutsch \cite{De91} showed that time travels to the past, including the above-mentioned scenario, would not lead to such problems if one just assumes that Nature imposes a \textit{causal consistency condition} that the state $x$ of the universe in the critical moment should be a fixed point of the operator $f$ describing the evolution in the CTC, i.e. that $x=f(x)$. In the Grandfather Paradox scenario, Nature would ``set" the state of the universe at time $A$ to a distribution where the traveler arrives with probability $\frac{1}{2}$ to keep things consistent, as a ``response" to the self-preventation action of the traveler. 

As Deutsch noted, a computer which sends part of its output back in time to  be used as part of its input can solve many computational problems much faster than what is believed to be possible without such loops. Bacon \cite{B04} showed that $\mathsf{NP}$-complete problems can be solved by polynomial-time computers with narrow CTCs. Aaronson and Watrous \cite{AW09} proved, as mentioned in the introduction, that 
$\mathsf{AC_{CTC}^{0}} = \mathsf{P_{CTC}} = \mathsf{BQP_{CTC}} = \mathsf{PSPACE_{CTC}} = \mathsf{PSPACE} $, where the subscript $\mathsf{CTC}$ under a class name indicates that the related machines have been reinforced by polynomial-width CTCs.

Let us review Aaronson and Watrous' model of quantum computation\footnote{The variant where the core model is classical is defined similarly.} with CTCs from \cite{AW09}: A deterministic polynomial-time algorithm $\mathcal{A}$ takes an input $w$, and prints the description of a quantum circuit $Q_{w}$ with rational amplitudes. $Q_{w}$ acts on two registers of polynomial size. One register holds the information that is sent from the end of the computation in the future through the CTC, whereas the other one is a standard causality-respecting register, including a bit that will be used to report the output of the computation. The circuit is executed, with the CTC register set by Nature to some state satisfying the causal consistency condition described above, and the causality-respecting register initialized to all zeros, and the result is read off the output bit. A language $L$ is said to be decided by such a CTC algorithm
$\mathcal{A}$ if all members of $L$ are accepted with high probability, and all nonmembers are rejected with high probability.

Note that this setup necessitates $\mathcal{A}$ to build a new CTC of the appropriate width for each different input $w$. This forces one \cite{De91} to take the cost of this construction into account when analyzing the complexity, and the resources required may well scale exponentially in the width of the CTC.\footnote{David Deutsch, personal communication.} The study of narrow CTCs, where this cost does not depend on the input length, is thus motivated.

\section{Our model} \label{section:ourmodel}
We will be considering several  computation models that are augmented with the capability of sending a single classical bit of information from the time of the end of their execution back to the beginning.

We define a machine of type $ \mathrm{M_{CTC_1}} $  as simply a machine of type $ \mathrm{M} $,\footnote{We will examine DFAs, PFAs, QFAs, both real-time and two-way DPDAs, time-bounded probabilistic and quantum Turing machines (PTMs and QTMs), and space-bounded deterministic Turing machines as core models. See \cite{GHI67,Si06,YS11A} for the standard definitions of these models.} which has access to an additional bit in a so-called \textit{CTC cell}. The CTC cell obtains its ``initial" distribution from the future, according to the causal consistency condition. The program format of an $ \mathrm{M_{CTC_1}} $ differs from that of an $ \mathrm{M} $ so that it specifies the transitions to be performed for all possible combinations of not only the input symbol, internal state, etc., but also the CTC cell value. The set $S$ of internal states of an $ \mathrm{M_{CTC_1}} $ is defined as the union of three disjoint sets $S_{n}$, $S_{p_{0}}$, and $S_{p_{1}}$. The states in $S_{n}$ are of the standard variety. When  they are entered, the states in $S_{p_{1}}$ ($S_{p_{0}}$) cause a 1 (0) to be sent back in time to be assigned to the CTC cell at the start of the execution. We assume that states in $S_{p_{1}} \cup S_{p_{0}}$ are entered only at the end of execution, (for real-time models, this is precisely when the machine is reading the end-marker symbol), and all states entered at that time are in $S_{p_{1}} \cup S_{p_{0}}$. Any number of members of $S$ can be designated as accept states. The input string $w$ is accepted if, for all stationary distributions of the evolution of the CTC bit induced by $w$, the machine accepts with probability at least $\frac{2}{3}$ with the CTC cell starting at that distribution. A string $w$ is rejected  if, for all stationary distributions of the evolution of the CTC bit induced by $w$, the machine rejects with probability at least $\frac{2}{3}$ with the CTC cell starting at that distribution. A language is recognized if all its members are accepted and all its nonmembers are rejected.

It is evident that any language recognized by any $ \mathrm{M_{CTC_1}} $ according to our definition is also decided by some CTC algorithm \textit{\'{a} la} the Aaronson-Watrous definition, described in Section \ref{section:Preliminaries}.
The motivation for the difference between the definitions is that the weakness of some of the core models we will use precludes us from performing any processing before using the CTC, and calculating or bounding the runtime, which determines the required ``range" of the CTC beforehand. For more on this issue, see Section \ref{section:feasibility}.

\section{Postselection and narrow CTCs} \label{section:pfaqfactc}

An important tool in the analysis of the capabilities of computers with narrow CTCs is the observation that one CTC bit endows any probabilistic or quantum core model with the power of postselection \cite{Aa05}. This fact is already known, but since we have not seen it stated explicitly anywhere, we present a demonstration of it below.

Postselection is the capability
of discarding all branches of a computation in which a specific event
does not occur, and focusing on the surviving branches for the final
decision about the membership of the input string in the recognized
language. A formal definition of polynomial time computation with postselection can be found in \cite{Aa05}, where it was proven that $\mathsf{PostBQP}$, the class of languages recognized by polynomial-time QTMs with postselection,  equals the class $\mathsf{PP}$. In the following, we present the analogous definition in the context of real-time computation with constant space.

A PFA (resp. QFA) $\mathcal{A}$ with postselection \cite{YS11B} is simply an ordinary PFA (QFA)  whose state set is partitioned into the sets of postselection accept, postselection reject, and nonpostselection states, and satisfies the condition that the probability that $\mathcal{A}$ will be in at least one postselection (accept or reject) state at the end of the processing is nonzero for all possible input strings. The overall acceptance and rejection probabilities of any input string $w$ ($P_{\mathcal{A}}^{a}(w)$ and $P_{\mathcal{A}}^{r}(w)$, respectively) are calculated
by simply discarding the computational paths ending at nonpostselection states, and performing a normalization so that the decision about the input is given by the postselection states:

\begin{equation}
\label{eq:postacc-postrej}
       P_{\mathcal{A}}^{a}(w) =
	\dfrac{p_{\mathcal{A}}^{a}(w)}{p_{\mathcal{A}}^{a}(w)+p_{\mathcal{A}}^{r}(w)}
	~~
	\mbox{ and }
	~~
	P_{\mathcal{A}}^{r}(w) =
	\dfrac{p_{\mathcal{A}}^{r}(w)}{p_{\mathcal{A}}^{a}(w)+p_{\mathcal{A}}^{r}(w)},
\end{equation}
where $p_{\mathcal{A}}^{a}(w)$ and $p_{\mathcal{A}}^{r}(w)$, respectively, are the acceptance and rejection probabilities of $w$ before the normalization.

\begin{lemma} \label{lemma:postsel}
Any language that can be recognized by a real-time automaton of type $ \mathrm{M} \in$ \{PFA,QFA\} with postselection can be recognized by an $ \mathrm{M_{CTC_1}} $.
\end{lemma}
\begin{proof} 
Let $\mathcal{A}$ be the given automaton with postselection. We construct an  $ \mathrm{M_{CTC_1}} $ $ \mathcal{A'} $. 
The values 1 and 0 of the CTC bit are associated with acceptance and rejection, as explained below.  $ \mathcal{A'} $ imitates the behavior of $\mathcal{A}$ until the end of the input. If $\mathcal{A}$ ends up at a postselection state, $ \mathcal{A'} $ halts with the same decision as $\mathcal{A}$ at that state. If $\mathcal{A}$ ends up at a nonpostselection state, $ \mathcal{A'} $ simply reports the value it sees in the CTC bit as its decision. $ \mathcal{A'} $ sends the value associated with its decision to the past as it halts.

The evolution of the CTC bit of $ \mathcal{A'} $ for input $w$ is described by the column stochastic matrix
\begin{equation}
	\left( \begin{array}{cc} 
		1-p_{\mathcal{A}}^{a}(w) & p_{\mathcal{A}}^{r}(w) \\
		p_{\mathcal{A}}^{a}(w) & 1-p_{\mathcal{A}}^{r}(w) 
		\end{array} \right),
\end{equation}
whose only stationary distribution is 
\begin{equation}
	\left( \begin{array}{c}
		\frac{p_{\mathcal{A}}^{r}(w)}{p_{\mathcal{A}}^{a}(w)+p_{\mathcal{A}}^{r}(w)} \\
		\frac{p_{\mathcal{A}}^{a}(w)}{p_{\mathcal{A}}^{a}(w)+p_{\mathcal{A}}^{r}(w)}
	\end{array} \right),
\end{equation}where the first and second entries stand for the probabilities of the values 0 and 1, respectively, meaning that $ \mathcal{A'} $ recognizes $\mathcal{A}$'s language with the same error probability.
\qed\end{proof}

Let $\mathsf{BPP_{CTC_1}}$ and $\mathsf{BQP_{CTC_1}}$ denote the classes of languages recognized by polynomial time PTMs and QTMs with narrow CTCs using classical bits, respectively.\footnote{Our definition of $\mathsf{BQP_{CTC_1}}$ is different from the $\mathsf{BQP_{CTC1}}$ given in \cite{AW09}, since Aaronson and Watrous consider a quantum bit sent through the CTC.} $\mathsf{BPP_{path}}$ is the class of languages recognized by polynomial-time PTMs with postselection.

The results $\mathsf{BPP_{path}} \subseteq \mathsf{BPP_{CTC_1}}$ and $\mathsf{PP} \subseteq \mathsf{BQP_{CTC_1}}$, that we alluded to in the introduction, are obtained using the link described above between CTCs and postselection.

We now present our main result that the power of postselection is all that a narrow CTC can confer on a computer. Let pPTM and pQTM denote polynomial-time PTMs and QTMs, respectively.

\begin{lemma} \label{lemma:selpos}
Any language that can be recognized by an $ \mathrm{M_{CTC_1}} $, where $ \mathrm{M} \in$ \{PFA,QFA,pPTM,pQTM\}, can be recognized by a machine of type $ \mathrm{M} $ with postselection.
\end{lemma}
\begin{proof} 
Let $ L $ be a language recognized by an $ \mathrm{M_{CTC_{1}}} $ named $ \mathcal{A} $. 
We start by constructing two machines of the standard type $ \mathrm{M} $, namely, $ \mathcal{A}_{0} $ and $ \mathcal{A}_{1} $, that simulate the computation of $ \mathcal{A} $ by fixing 0 and 1 for the value of the CTC bit, respectively.

Let $ p_{\mathcal{A}_{i}}^{j} (w) $ denote the probability that $ \mathcal{A}_{i} $ will reach a configuration corresponding to sending the bit $j$ to the past at the end of its computation when started on input $w$, where $ i,j \in \{0,1\} $. The CTC bit's evolution is described by

\begin{equation}
	\left( \begin{array}{cc} 
		1- p_{\mathcal{A}_{0}}^{1} (w)  &  p_{\mathcal{A}_{1}}^{0} (w)  \\
		 p_{\mathcal{A}_{0}}^{1} (w)  & 1- p_{\mathcal{A}_{1}}^{0} (w)  
		\end{array} \right),
\end{equation}
with stationary distribution
\begin{equation}
	\left( \begin{array}{c}
		\frac{ p_{ \mathcal{A}_{1} }^{0}(w) }{p_{ \mathcal{A}_{0} }^{1}(w) + 
	p_{ \mathcal{A}_{1} }^{0}(w)} \\
		\frac{ p_{ \mathcal{A}_{0} }^{1}(w) }{p_{ \mathcal{A}_{0} }^{1}(w) + 
	p_{ \mathcal{A}_{1} }^{0}(w)}
	\end{array} \right).
\end{equation}

For $ w \in L $, we therefore have the  inequality
\begin{equation}
	\label{eq:overall-acc}
	\frac{ p_{ \mathcal{A}_{1} }^{0}(w) }{p_{ \mathcal{A}_{0} }^{1}(w) + 
	p_{ \mathcal{A}_{1} }^{0}(w)}
	p_{ \mathcal{A}_{0} }^{a}(w) +
	\frac{ p_{ \mathcal{A}_{0} }^{1}(w) }{p_{ \mathcal{A}_{0} }^{1}(w) + 
	p_{ \mathcal{A}_{1} }^{0}(w)}
	p_{ \mathcal{A}_{1} }^{a}(w)
	\geq \frac{2}{3}.
\end{equation}

We claim that, for all machine types $ \mathrm{M} $ mentioned in the theorem statement, one can construct an instance of $ \mathrm{M} $, say, $ \mathcal{A}' $, which will have two mutually exclusive collections of states, say, $ S_{a} $ and $ S_{r} $, such that the probability that $ \mathcal{A}' $ halts in $ S_{a} $ when started on input $w$ is
\begin{equation}
 \label{eq:msacc}
	\frac{1}{2} \left( p_{ \mathcal{A}_{1} }^{0}(w) p_{ \mathcal{A}_{0} }^{a}(w) +
	p_{ \mathcal{A}_{0} }^{1}(w) p_{ \mathcal{A}_{1} }^{a}(w) \right),
\end{equation}
and the probability that $ \mathcal{A}' $ halts in $ S_{r} $ is
\begin{equation}
	\label{eq:msrej}
	\frac{1}{2} \left( p_{ \mathcal{A}_{1} }^{0}(w) \left( 1 - p_{ \mathcal{A}_{0} }^{a}(w) \right) +
	p_{ \mathcal{A}_{0} }^{1}(w) \left( 1 - p_{ \mathcal{A}_{1} }^{a}(w) \right) \right).
\end{equation}
For instance, if $ \mathrm{M} $=pPTM, we first build two pPTMs, say, $ \mathcal{A}_{10} $ and $ \mathcal{A}_{01} $, for handling the two operands of the addition in Equation \ref{eq:msacc} by sequencing $ \mathcal{A}_{0} $ and $ \mathcal{A}_{1} $ to run on the input in the two different possible orders. $ \mathcal{A}' $ simply runs $ \mathcal{A}_{10} $ and $ \mathcal{A}_{01} $, with probability $\frac{1}{2}$ each. Equation \ref{eq:msrej} is handled similarly.

If $ \mathcal{A} $ is a real-time machine, the sequential processing described above is not permitted, and we instead perform tensor products of $ \mathcal{A}_{0} $ and $ \mathcal{A}_{1} $ to obtain the submachines $ \mathcal{A}_{10} $ and $ \mathcal{A}_{01} $.

Once $ \mathcal{A}' $ is completed, we view it as a machine with postselection, by postselecting on the states in $ S_{a} \cup S_{r} $ being reached at the end. $ S_{a} $ will be designated to be the set of accept states. $ S_{r} $ will be the reject states. 

Performing the normalization described in Equation \ref{eq:postacc-postrej} on Equations \ref{eq:msacc} and \ref{eq:msrej} to calculate $ \mathcal{A}' $'s probability of acceptance, one obtains precisely  the expression depicted in Equation \ref{eq:overall-acc}.
The case of $ w \notin L $ is symmetric. We conclude that $ \mathcal{A}' $ recognizes $L$ with exactly the same error probability as $ \mathcal{A} $.
\qed\end{proof}

We have therefore proven
\begin{theorem}
For any $ \mathrm{M} \in$ \{PFA,QFA,pPTM,pQTM\}, $ \mathrm{M_{CTC_1}} $ is equivalent in language recognition power to a machine of type $ \mathrm{M} $ with postselection.
\end{theorem}
\begin{corollary}
$\mathsf{BPP_{path}} = \mathsf{BPP_{CTC_1}}$, and $\mathsf{PP} = \mathsf{PostBQP} = \mathsf{BQP_{CTC_1}}$.
\end{corollary}

We can use Lemma \ref{lemma:postsel} to demonstrate the superiority of real-time PFAs and QFAs with narrow CTCs over their standard versions, which can only recognize regular languages with bounded error:  
For a given string $ w $, let $ |w|_{\sigma} $ denote the number of occurrences of symbol $ \sigma $ in $ w $. 
The nonregular language $ L_{eq} = \{ w \in \{a,b\}^{*} \mid |w|_{a} = |w|_{b} \} $ can be recognized by a PFA with postselection \cite{YS11B}. As for quantum machines, the language
$ L_{pal} = \{w \in \{a,b\}^{*} \mid w = w ^{r} \} $ is recognized by a QFA with postselection \cite{YS11B}. 
$ L_{eq} $ is known \cite{Fr81} to be recognizable by two-way PFAs at best in superpolynomial time \cite{GW86}, and the best known two-way QFA algorithm \cite{AW02} for $ L_{pal}$ has exponential expected runtime. Furthermore, $ L_{pal}$ is known \cite{DS92} to be unrecognizable by even two-way PFAs with bounded error, and no PFA with postselection can outperform a standard two-way PFA \cite{YS11B}, so we have established that finite-state quantum models with narrow CTCs outperform their probabilistic counterparts:
\begin{corollary}
	The class of languages recognized by $ QFA_{CTC_1} $s properly contains the class of languages recognized by $ PFA_{CTC_1} $s.
\end{corollary}

\section{Weak deterministic models with narrow CTCs} \label{section:detctc}

We adapt the argument used in \cite{AW09} to prove that $ \mathsf{P_{CTC}} \subseteq \mathsf{PSPACE} $ to state the following upper bounds for the powers of deterministic machines with narrow CTCs:

\begin{theorem} \label{theorem:detub}
	Let $ \mathcal{A} $ be any machine of type $ \mathrm{M_{CTC_1}} $ that recognizes a language according to the 	definition in Section \ref{section:ourmodel}, where $ \mathrm{M} $ is a deterministic model. $ \mathcal{A} $ can be simulated by running two machines of type $ \mathrm{M} $ in succession.
\end{theorem}
\begin{proof} Assume that $ \mathcal{A} $ never enters an infinite loop. (We  get rid of this assumption in Section \ref{section:feasibility}.) For a given input string $w$, let $f_{w}$ be the mapping among probability distributions over the CTC bit realized by  $ \mathcal{A} $ when running on $w$.  We have to find a distribution 
$ d $, such that $f_{w}(d)=(d)$, and see how $ \mathcal{A} $ responds to $w$ when it starts with the CTC bit having that distribution.

We first run a machine $ \mathcal{A}_{1} $ of type $ \mathrm{M} $ obtained by fixing the CTC bit of $ \mathcal{A} $ to 0 on input $w$. If this run ends at a state belonging to  $S_{p_0}$, the collection of $ \mathcal{A} $'s states that set the CTC bit to 0, we have found that $ \left( \begin{array}{c} 1 \\ 0 \end{array} \right) $ is a stationary distribution, and the response given to $w$ by 
$ \mathcal{A}_{1} $ is what $ \mathcal{A} $ itself would announce if it were executed. If this first stage ends within $S_{p_1}$, then we run another machine $ \mathcal{A}_{2} $ of type $ \mathrm{M} $, obtained by fixing the CTC bit of $ \mathcal{A} $ to 1, on input $w$. Note that the only remaining possibilities for stationary distributions at this stage are $ \left( \begin{array}{c} 0 \\ 1 \end{array} \right) $ and $ \left( \begin{array}{c} \frac{1}{2} \\ \frac{1}{2} \end{array} \right) $, and in either case, $ \mathcal{A}_{2} $'s response to $w$ is certain to be identical to $ \mathcal{A} $'s response, since $ \mathcal{A} $ cannot have an error probability as big as $\frac{1}{2}$.
\qed \end{proof}

This construction can be realized by a two-way version of model $ \mathrm{M} $. 
It is known \cite{Sh59} that two way DFA's are equivalent to their one-way versions.
\begin{corollary}
	One-way $ DFA_{CTC_1} $s recognize precisely the regular languages.
\end{corollary}

Two-way DPDAs are more powerful than one-way DPDAs \cite{GHI67}. 
Given a one-way $DPDA_{CTC_1}$ $ \mathcal{A} $, we can apply the idea of Theorem \ref{theorem:detub} to obtain three DPDAs as follows:
$ \mathcal{A}_{1} $ is obtained by fixing the CTC bit to 0, and accepting if computation ends in a member of $S_{p_0}$ that is also an accept state.
$ \mathcal{A}_{2} $ is obtained by fixing the CTC bit to 0, and accepting if computation ends in any member of $S_{p_1}$. 
$ \mathcal{A}_{3} $ is obtained by fixing the CTC bit to 1, with no change to the accept states of $ \mathcal{A} $.

Calling the languages recognized by these three machines $ L_1 $, $ L_2 $, and $ L_3 $, respectively, it is easy to see that the language recognized by $ \mathcal{A} $ is $ L_1 \cup (L_2 \cap L_3)$. We conclude that any language recognized by a one-way $DPDA_{CTC_1}$ can be expressed as the union of a deterministic context-free language (DCFL) with a language that is the intersection of two DCFLs.

To demonstrate that $DPDA_{CTC_1}$s are actually more powerful than ordinary DPDAs, we will show that the capability of sending a finite number of bits to the past endows a machine with the power of limited nondeterminism. 

The amount of nondeterminism used by a PDA can be quantified in the following manner \cite{He97}: The branching of a single move of a (nondeterministic) PDA is defined as the number of next configurations that are possible from the given configuration. The branching of a computation path of a PDA $ \mathcal{N} $ is the product of the branchings of all the moves in this path. The branching of a string $ w $ accepted by $ \mathcal{N} $ is the minimum of the branchings of the paths of $ \mathcal{N} $ that accept $w$. Finally, the branching of $ \mathcal{N} $ is the maximum of the branchings of the strings accepted by $ \mathcal{N} $. 

\begin{fact}
	\label{fact:brun}
	The class of languages recognized by PDAs with branching $ k $ is the class of unions of $ k $ DCFLs.
\end{fact}

\begin{theorem} 
	\label{theorem:nondet}
	Any language that can be expressed as the union of two DCFLs can be recognized by a one-way $DPDA_{CTC_1}$.
\end{theorem}
\begin{proof} 
By Fact \ref{fact:brun}, we only need to show how to build a $DPDA_{CTC_1}$ that can simulate a given PDA $ \mathcal{N} $ with branching 2. 

We convert $ \mathcal{N} $ to an equivalent PDA $ \mathcal{N'} $, all of whose computational paths have branching exactly 2, by modifying the program so that for any computational path of $ \mathcal{N} $ with branching greater than 2, $ \mathcal{N'} $ simply scans the input until the end, and rejects without performing that excess branching. For every nondeterministic state of $ \mathcal{N'} $, name the two outgoing branches 0 and 1. Convert $ \mathcal{N'} $ to a $DPDA_{CTC_1}$ $ \mathcal{A} $ which simulates $ \mathcal{N'} $, selecting the \textit{i}th branch if and only if it sees the value \textit{i} in the CTC bit. At the end of the input, $ \mathcal{A} $ sends the name of the current branch to the past if it is accepting the input. It sends the name of the other branch otherwise.

\begin{table}[h]
	\scriptsize
	\caption{Evolutions and stationary distributions of the CTC bit of $ \mathcal{A} $}
	\centering
	\begin{tabular}{|c|c|c|c|}
		\hline
		$ \mathsf{branch}_{0} $ & $ \mathsf{branch}_{1} $ & CTC transformation & Stationary distribution 
		\\ \hline
		Acc & Acc & $ \left( \begin{array}{cc} 1 & 0 \\ 0 & 1 \end{array} \right) $ & Any distribution
		\\ \hline
		Acc & Rej & $ \left( \begin{array}{cc} 1 & 1 \\ 0 & 0 \end{array} \right) $ & 
		$ \left( \begin{array}{c} 1  \\ 0 \end{array} \right) $
		\\ \hline
		Rej & Acc & $ \left( \begin{array}{cc} 0 & 0 \\ 1 & 1 \end{array} \right) $ & 
		$ \left( \begin{array}{c} 0 \\ 1 \end{array} \right) $
		\\ \hline
		Rej & Rej & $ \left( \begin{array}{cc} 0 & 1 \\ 1 & 0 \end{array} \right) $ & 
		$ \left( \begin{array}{c} \frac{1}{2} \\ \frac{1}{2} \end{array} \right) $
		\\ \hline
	\end{tabular}
	\label{tab:DPDA}
\end{table}

We consider the four possible cases corresponding to accept/reject responses of the two paths, and the associated evolutions of the CTC bit in Table \ref{tab:DPDA}. It is evident that $ \mathcal{A} $ recognizes the language of $ \mathcal{N} $ with zero error.
\qed\end{proof}

Since there exist languages (e.g. $\{a^{i}b^{j}c^{k} | i=j$ or $i=k\}$) that are not deterministic context-free, but which can be expressed as the union of two DCFLs, we conclude that the computation power of DPDAs is actually increased by the addition of a narrow CTC.

\section{Implementation of arbitrary negative delays using a constant-length channel} \label{section:feasibility}

In our model of computation with CTCs, the runtime, (that is, the length of the backward jump in time that the CTC bit will go through,) is not known at the start of the computation.  In the Aaronson-Watrous model, this ``range" is known at the end of the preprocessing stage prior to the point where the information arrives through the CTC, so one can view their program as building a mechanism that can transport a specified quantity of information with that fixed range, and then using it. Let us show that our setup can be implemented using a time machine of the Morris-Thorne-Yurtsever variety\footnote{See the brief description in Section \ref{section:Preliminaries}.} of constant range $ T $ seconds. 

Assume that all instructions in our programming language take an equal amount of time to be executed, and that a backward jump of $ T $ seconds amounts to a difference of $ k>3 $ instructions, in the sense that if the \textit{i}th instruction sends a bit backwards, that bit will be available to the (\textit{i-k})th, but not to the (\textit{i-k-1})th instruction.  

Consider the  trace of a computation conforming to our definition of Section \ref{section:ourmodel} in Figure \ref{fig:variable-length-channel}. In the figure, the first instruction $ r $ indicates the point where the CTC bit coming from the future is received, and $ s $ indicates the point where it is sent. (For simplicity,  assume that the program idles (performs a no-op) whenever an $ r $ ``instruction" is executed in this and the following figure.)
The idea in transforming this program to one which uses a fixed-length channel to the past is simply for the information to travel from the end of the execution back to its start in several ``hops," rather than one. For this purpose, we rewrite our programs so that they send back the current value of the CTC bit backwards once every $ k-1 $ instructions.  The modified program would then have a computation of the kind depicted in Figure \ref{fig:fixed-length-channel}.

\begin{figure}
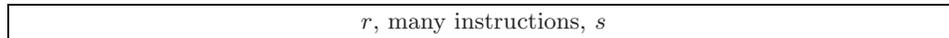

	\begin{center}
	\fbox{
		\begin{minipage}{\textwidth}
			\centering
			$ r $, many instructions, $ s $
		\end{minipage}
	}
	\end{center}
	\caption{Trace of a computation with a variable-length channel}
	\label{fig:variable-length-channel}
\end{figure}

\begin{figure}
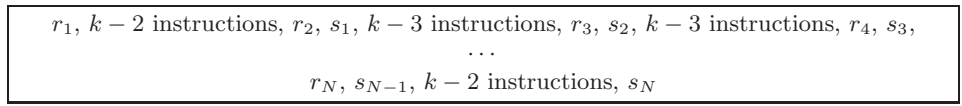

	\begin{center}
	\fbox{
		\begin{minipage}{\textwidth}
			\centering
			$ r_{1} $, $ k-2 $ instructions, $ r_{2} $, $ s_{1} $, $ k-3 $ instructions, 
			$ r_{3} $, $ s_{2} $, $ k-3 $ instructions, $ r_{4} $, $ s_{3} $, 
			\\ $ \cdots $ \\
			$ r_{N} $, $ s_{N-1} $,  $ k-2 $ instructions, $ s_{N} $
		\end{minipage}
	}
	\end{center}
	\caption{Trace of a computation with a fixed-length channel}
	\label{fig:fixed-length-channel}
\end{figure}

In Figure \ref{fig:fixed-length-channel}, the $ r_{i}  $-$ s_{i} $ pairs indicate the corresponding arrival and departure points of the CTC bit. The modification to be performed on the program is clearly achievable by a finite-state mechanism. Note that the first and last ``segments" (i.e. sequences of standard instructions, not involving time travel) of the execution are one step longer than the intermediate segments, since  we do not need send or receive to or from time points outside the duration of the execution. Also note that the  last segment has to be padded with no-ops if necessary\footnote{This necessity, and the inclusion of the intermediate $r$ and $s$ instructions that do not consume input symbols, mean that the resulting program is not real-time, although we can insert paddings in our language definitions of earlier sections to obtain sister languages that can be recognized in real-time by the approach of this section.} to make sure that $ r_{N} $ and $ s_{N} $ are separated with the proper length of time. 

We show that the fixed- vs. variable-length channel setups are equivalent, by demonstrating that the application of the causal consistency condition to all $ r_{i}  $-$ s_{i} $ pairs of Figure \ref{fig:fixed-length-channel} yields the same restriction on stationary distributions as the single application in Figure \ref{fig:variable-length-channel}.

We will assume, without loss of generality, that our original $M_{CTC_1}$ program consults its CTC bit only at the beginning of its computation. Note that this program can now be considered to consist of two branches $ b_{0} $ and $ b_{1} $, differentiated by the initial CTC value. Let $p^{i}$ denote the probability that a 0 will be sent back at the end of the execution, given that branch $ b_{i} $ has been selected at the beginning. It is important to see that these values are unchanged by a switch between the two setups.

Consider all pairs of the form $(r_i$-$s_i, r_{i+1}$-$s_{i+1})$, where $1<i<N$ in one of these two branches. Let the stationary distribution associated with the $r_i$-$s_i$ pair be $ \left( \begin{array}{c} p_i \\ 1-p_i \end{array} \right) $. This must equal the stationary distribution of the $r_{i+1}$-$s_{i+1}$ pair, since the value sent back by $s_i$ equals the one received by $r_{i+1}$: The evolution in the  $r_i$-$s_i$ pair is
\begin{equation}
	\left(
	\begin{array}{cc}
		p_{i+1} & p_{i+1} \\
		1-p_{i+1} & 1-p_{i+1}
	\end{array}
	\right),
\end{equation}
whose only stationary distribution is indeed $\left( \begin{array}{c} p_{i+1} \\ 1-p_{i+1} \end{array} \right) $.

Now consider the evolution in the $ r_1 $-$ s_1 $ pair. Using superscripts again to denote the probabilities associated with the two branches, this is
\begin{equation}
	\left(
	\begin{array}{cc}
		p^{0}_2 & p^{1}_2 \\
		1-p^{0}_2 & 1-p^{1}_2
	\end{array}
	\right),
\end{equation}
which, by the fact that $p^{0}_2=p^{0}_N$ and $p^{1}_2=p^{1}_N$, equals
\begin{equation}
	\left(
	\begin{array}{cc}
		p^{0}_N & p^{1}_N \\
		1-p^{0}_N & 1-p^{1}_N
	\end{array}
	\right).
\end{equation}

Since  $p^{0}_N$ and $p^{1}_N$ are exactly the probabilities of 0 being assigned to the CTC bit by the respective branches in the variable-length channel setup as well, the matrix above is precisely the evolution matrix induced by that program running on the same input. 

One useful consequence of this result is that it enables us to consider core models in which the machine can have multiple computation paths of different length in a single run. (Note that the models of Sections \ref{section:pfaqfactc} and \ref{section:detctc}  avoided this issue, since they were either restricted to real-time processing, or had sufficient resources to ensure that all paths have the same runtime, with any premature acceptances or rejections postponed until the scanning of the right input-end-marker.) In the Aaronson-Watrous model, this possibility can be handled by effectively making sure, by adding no-ops to short paths, that all paths have the same runtime, and building the CTC according to that specification. For the weak models we are considering, it is sufficient to view the CTC-assisted computation as being implemented with fixed-length channels. Note that the two branches of the computation need not be of equal length for the reasoning above to hold. This will help us characterize the power of certain models with two-way access to the input, as soon as we clarify the issue of the possibility of some branches entering infinite loops.  

What happens if a $ \mathrm{M_{CTC_1}}$ program (for a deterministic $ \mathcal{M} $) employing a fixed-length channel enters an infinite loop on, say, branch 0, never reaching accept or reject states? It is of course more realistic to assume that in this case the execution will finish, (possibly a very long time later), due to external reasons. In any case, no new value will be assigned to the CTC bit on that branch. The reasoning above about all the $p^{0}_{i}$ equaling each other for $i>1$ still applies, but since no assignment is performed within the last hop, consideration of the end of the branch does not help us to constrain this probability value. Assuming that branch 1 halts, sending 0 to the past with probability $p^{1}$, the evolution within the $ r_{1} $-$ s_{1} $ period is
\begin{equation}
	\left(
	\begin{array}{cc}
		1 & p^{1} \\
		0 & 1-p^{1}
	\end{array}
	\right).
\end{equation}
$ \left( \begin{array}{c} 1 \\ 0 \end{array} \right) $
is a stationary distribution of this matrix. The program neither accepts nor rejects when the CTC bit starts with this distribution, and therefore such programs, i.e. those with infinite loops, do not contribute to the class of languages recognized by $ \mathrm{M_{CTC_1}}$s. 
Since $ \mathrm{M_{CTC_1}}$ programs with finite-length branches can be handled using fixed-length channels, and simulated using the technique of Theorem \ref{theorem:detub}, we have

\begin{corollary}
For any space bound s, $\mathsf{SPACE}(s)_{\mathsf{CTC_1}}=\mathsf{SPACE}(s)$.
\end{corollary}

\begin{corollary}
The computational power of two-way DPDAs does not change with the addition of a narrow CTC.
\end{corollary}

\section{Concluding remarks and open questions} \label{section:Conclusion}

We have examined the  power of several computational models augmented by the capability of sending a single classical bit to the past. We have characterized the classes $\mathsf{BPP_{CTC_1}}$ and $\mathsf{BQP_{CTC_1}}$ in terms of classical conventional classes, and shown that real-time probabilistic and quantum finite automata, as well as deterministic pushdown automata,  benefit from narrow CTCs. In Section \ref{section:feasibility}, we establish that CTCs remain useful even if the information channel to the past has a small fixed range, e.g. a few seconds, and that narrow CTCs do not change the power of deterministic models with two-way access to the input string. One-way deterministic finite automata are also not affected. 

In an earlier paper \cite{YS11B}, we had shown that machines with postselection have precisely the same power as conventional machines that are able to reset their input head to the start of the input string and switch to the initial state, to restart the computation all over again \cite{YS10B}. The new link to narrow CTCs shows that postselection is indeed a profoundly interesting concept that requires further investigation.

Some open questions remain. 
Can we pin down the power of $DPDA_{CTC_1}$s further than we have done in Theorems \ref{theorem:detub} and \ref{theorem:nondet}?
Most of our results are obviously generalizable to CTCs with capacities of $k>1$ bits. With more CTC bits, one can clearly implement more nondeterministic choices. With more nondeterminism, one can obtain more succinct finite automata for certain languages, and build PDAs with superior language recognition capability \cite{He97}. It would be interesting to clarify and quantify these relationships.

In most of our programs, for instance, the machines of Lemma \ref{lemma:postsel}, the machine's decision about the input can be read off the CTC bit immediately after the start of the execution with high probability, without the need to wait for the whole input string to be scanned. What stops us from ``cheating," turning the computer off and using the output ``for free," at that point? The answer, according to the analysis in \cite{De91}, is that the CTC bit will be set to the desired disribution only if nothing interferes with the complete scheduled run of the machine, and its value will be ``corrupted" if, say, a meteorite is likely to hit the computer in the middle of the computation, or if the molecules in our brains are configured so that we are likely to cheat. This requirement of perfect isolation from external interference during the computation is yet another interesting aspect of CTC-assisted computation.
\section*{Acknowledgements} \label{section:Acknowledgements}

We thank David Deutsch, Scott Aaronson, Amos Ori, and Taylan Cemgil for their helpful answers to our questions.

\bibliographystyle{alpha}
\bibliography{YakaryilmazSay}

\end{document}